\documentclass[twocolumn, superscriptaddress,
showpacs,preprintnumbers,amsmath,amssymb]{revtex4}

\usepackage{amssymb,amsmath,amsthm}
\usepackage{epsfig}
\usepackage{multirow}
\newtheorem{Thm}{Theorem}

\newtheorem{Lem}[Thm]{Lemma}
\newtheorem{Prop}{Proposition}

\theoremstyle{definition}

\newcommand{\bra}[1]{{\left\langle #1 \right|}}
\newcommand{\ket}[1]{{\left| #1 \right\rangle}}

\newcommand{\T}{\mbox{$\mathrm{tr}$}}

\begin{document}
\title{Negativity and strong monogamy of multi-party quantum entanglement beyond qubits}

\author{Jin Hyuk Choi}
\affiliation{
 Humanitas College, Kyung Hee University, Yongin-si, Gyeonggi-do 446-701, Korea
}

\author{Jeong San Kim}
\email{freddie1@khu.ac.kr} \affiliation{
 Department of Applied Mathematics and Institute of Natural Sciences, Kyung Hee University, Yongin-si, Gyeonggi-do 446-701, Korea
}
\date{\today}

\begin{abstract}
We propose the square of convex-roof extended negativity(SCREN) as a powerful candidate to characterize
strong monogamy of multi-party quantum entanglement.
We first provide a strong monogamy inequality of multi-party entanglement using SCREN, and show
that the tangle-based multi-qubit strong monogamy inequality can be rephrased by SCREN.
We further show that SCREN strong monogamy inequality is still true for the counterexamples that violate
tangle-based strong monogamy inequality in higher-dimensional quantum systems rather than qubits.
We also analytically show that SCREN strong monogamy inequality is true for a large class of multi-qudit states,
a superposition of multi-qudit generalized W-class states and vacuums.
Thus SCREN is a good alternative to characterize the strong monogamy of entanglement even in multi-qudit systems.
\end{abstract}

\maketitle

\section{Introduction}
\label{Intro}
Quantum entanglement is a quantum correlation used as a resource in various applications
of quantum information theory such as quantum teleportation and quantum cryptography~\cite{tele, qkd1, qkd2}.
One important property of entanglement is its restricted shareability in multi-party quantum systems, which does not have
any classical counterpart. This restriction of entanglement shareability among multi-party systems is
known as the {\em monogamy of entanglement}~(MoE)~\cite{CKW, OV, KW, T04, KGS}.

The first mathematical characterization of MoE was established by
Coffman-Kundu-Wootters(CKW) as an inequality~\cite{CKW}; for a three-qubit
pure state $\ket{\psi}_{ABC}$,
$$\tau\left(\ket{\psi}_{A|BC}\right) \geq \tau\left(\rho_{A|B}\right)+\tau\left(\rho_{A|C}\right),$$
where  $\tau\left(\ket{\psi}_{A|BC}\right)$ is the {\em one-tangle} of $\ket{\psi}_{ABC}$ quantifying the pure state entanglement between $A$ and $BC$,
and $\tau\left(\rho_{A|B}\right)$ (similarly with $\tau\left(\rho_{A|C}\right)$) is the {\em two-tangle}
of the reduced density matrix $\rho_{AB}=\T_{C}\ket{\psi}_{ABC}\bra{\psi}$
quantifying the two-qubit entanglement inherent in $\rho_{AB}$.

This inequality is also referred to as {\em CKW inequality}, and it shows the mutually exclusive nature of
two-qubit entanglement shared in three-qubit systems;
more entanglement shared between two qubits $A$ and $B$
leads to
less entanglement between the other two qubits $A$ and $C$
so that their summation does not exceed the total entanglement between $A$ and $BC$.
Moreover, the residual entanglement from the difference between left and right-hand sides of CKW inequality
is interpreted as the genuine three-qubit entanglement, {\em three-tangle}.

Later, CKW inequality was generalized for multi-qubit systems~\cite{OV} as well as some cases of higher-dimensional
quantum systems~\cite{KDS, KSRenyi, KT, KSU}. A general monogamy
inequality for arbitrary quantum systems was established in terms of the {\em squashed entanglement}~\cite{CW04, BCY10}.

Recently, the definition of three-tangle was generalized into arbitrary $n$-qubit systems, namely {\em $n$-tangle}
quantifying the genuine multi-qubit entanglement. By conjecturing the nonnegativity of
the $n$-tangle, the concept of {\em strong monogamy}(SM) inequality of $n$-qubit entanglement was proposed~\cite{LA}.
Although an analytical proof of SM conjecture for arbitrary multi-qubit states seems to be a formidable challenge due
to the numerous optimization processes arising in the definition of $n$-tangle, an extensive numerical evidence was presented
for four qubit systems together with an analytical proof for some cases of multi-qubit systems~\cite{LA, Kim14}.

However, tangle is known to fail in the generalization of CKW inequality for higher dimensional quantum systems
rather than qubits; there exist quantum states in $3 \otimes 3\otimes 3$ and even in $3 \otimes 2\otimes 2$ quantum systems violating
CKW inequality~\cite{OU, KS}. Because SM inequality proposed in~\cite{LA} is reduced to CKW inequality for $n=3$,
these counterexamples of CKW inequality also implies the violation of SM inequality using tangles in higher-dimensional systems rather than qubits.

Here we propose the {\em square of convex-roof extended negativity}(SCREN) as a powerful candidate to characterize
the strongly monogamous property of multi-qudit systems.
We first provide a SM inequality of multi-party entanglement using SCREN, and show that the SM inequality of multi-qubit entanglement using tangle~\cite{LA, Kim14}
can be rephrased by SCREN. This SCREN SM inequality is also true for the counterexamples of tangle in higher-dimensional systems.
Moreover, we analytically show that SCREN SM inequality is saturated by a large class of multi-qudit states, a superposition of multi-qudit generalized W-class states and vacuums.
Thus SCREN is a good alternative for strong monogamy of multi-party entanglement even in higher-dimensional systems.

The paper is organized as follows. In Sec.~\ref{sub: Negativity}, we review the definition of negativity, and provide the relation between tangle and SCREN
for multi-qubit monogamy inequality in Sec.~\ref{sub: mono qubit}. In Sec.~\ref{sub: qubitSM}, we recall the multi-qubit SM inequality in terms of tangle, and
propose a multi-qudit SM inequality using SCREN in~\ref{Subsec: SCRENmono}.
In Sec.~\ref{Subsec: gW}, we provide the definition of multi-qudit generalized W-class states as well as some useful properties of this class of states.
In Sec.~\ref{subsec: Sat}, we analytically show that the SCREN SM inequality of multi-qudit entanglement is saturated by a superposition of generalized
W-class states and vacuum. In Sec.~\ref{Sec: Conclusion}, we summarize our results.

\section{Negativity and Monogamy of Multi-Party Quantum Entanglement}
\label{Sec: N_mono}

\subsection{Negativity}
\label{sub: Negativity}

For a bipartite pure state $\ket{\phi}_{AB}$ in a $d\otimes d'$
($d\le d'$) quantum system with its Schmidt decomposition,
\begin{equation}
\ket{\phi}_{AB}~=~\sum_{i=0}^{d-1} \sqrt{\lambda_{i}}\ket{ii},~~\lambda_{i} \geq 0,~\sum_{i=0}^{d-1}\lambda_{i}~=1,
\label{Schmidt}
\end{equation}
its {\em negativity} is defined as
\begin{align}
\mathcal{N}(\ket{\phi}_{A|B})=\left\|\ket{\phi}_{AB}\bra{\phi}^{T_B}\right\|_1-1
=2\sum_{i<j}\sqrt{\lambda_{i}\lambda_{j}},
\label{eq:pure_negativity}
\end{align}
where
\begin{align}
\ket{\phi}_{AB}\bra{\phi}^{T_B} = \sum_{i,j=0}^{d-1}\sqrt{\lambda_{i}\lambda_{j}}\ket{ij}_{AB}\bra{ji}
\label{eq:pt state}
\end{align}
is the partial transposition of $\ket{\phi}_{AB}$ and
$\left\|\cdot\right\|_1$ is the trace norm~\cite{VidalW}.

Because the possible negative eigenvalues of the partially transposed state in Eq.~(\ref{eq:pt state}) are
$- \sqrt{\lambda_{i}\lambda_{j}}$ for $i < j$ with corresponding eigenstates
$\ket{\psi_{ij}}_{AB}~=~\frac{1}{\sqrt{2}}(\ket{ij}_{AB}-\ket{ji}_{AB})$, the definition of negativity in Eq.~(\ref{eq:pure_negativity})
is thus the sum of all possible negative eigenvalues with a constant proportion~\cite{Negativity}.
Eq.~(\ref{eq:pure_negativity}) can also have an alternative definition as
\begin{align}
\mathcal{N}(\ket{\phi}_{A|B})=2\sum_{i<j}\sqrt{\lambda_{i}\lambda_{j}}
=(\T{ \sqrt{\rho_A}})^2 -1,
\label{eq:pure_negativity2}
\end{align}
where $\rho_A = \T_B{\ket{\phi}_{AB}\bra{\phi}}$ is the reduced density matrix of $\ket{\phi}_{AB}$ on subsystem $A$.
For a bipartite mixed state $\rho_{AB}$, its negativity is analogously defined as
\begin{equation}
\mathcal{N}(\rho_{A|B})=\left\|{\rho_{AB}}^{T_B}\right\|_1-1,\label{eq:negativity}
\end{equation}
where $\rho_{AB}^{T_B}$ is the partial transposition of $\rho_{AB}$.

{\em Positive partial transposition}(PPT)~\cite{Peres, Horodeckis1} gives a
separability criterion for bipartite pure states and two-qubit mixed states.
PPT is also a necessary and sufficient condition for nondistillability in
$2\otimes n$ quantum system \cite{Horodecki1,DCLB}.
However, there also exist entangled mixed states with
PPT in higher-dimensional quantum systems rather than $2\otimes 2$ or
$2\otimes 3$ quantum systems.~\cite{Horodecki1,Horodeckis2}.
For this case, negativity in Eq.~(\ref{eq:negativity})
cannot distinguish PPT bound entangled states from separable
states, and thus, negativity itself is not sufficient to be a good
measure of entanglement even in a $2\otimes n$ quantum system.

One way to overcome this rack of separability criterion of negativity in higher-dimensional mixed quantum states
is using {\em convex-roof extension}~\cite{LCOK};
for a bipartite mixed state mixed state $\rho_{AB}$, its convex-roof extended negativity is
\begin{equation}
\mathcal{N}_m(\rho_{A|B})=\min_{\{p_k, \ket{\phi_k}\}} \sum_k p_k
\mathcal{N} (\ket{\phi_k}_{A|B}), \label{eq:c_negativity}
\end{equation}
where the minimum is taken over all possible pure state
decompositions of $\rho_{AB}={\sum_k p_k \ket{\phi_k}_{AB}\bra{\phi_k}}$.
Convex-roof extended negativity gives a perfect discrimination of PPT bound entangled states and separable states
in any bipartite quantum system. Moreover, it was also shown that the quantity in Eq.~(\ref{eq:c_negativity}) cannot be
increased by local quantum operations and classical communications(LOCC)~\cite{KDS, LCOK}.

\subsection{Monogamy Inequality Using Negativity}
\label{sub: mono qubit}

For a two-qubit pure state $\ket{\psi}_{AB}$~\cite{schtau}, its tangle (or one-tangle) is defined as
\begin{equation}
\tau\left(\ket{\psi}_{A|B}\right)=4\det \rho_A,
\label{1tangle}
\end{equation}
with the reduced density matrix $\rho_{A}=\T_{B}\ket{\psi}_{AB}\bra{\psi}$.
For a two-qubit mixed state $\rho_{AB}$, its tangle (or two-tangle) is defined as
\begin{equation}
\tau\left(\rho_{A|B}\right)=\bigg[\min_{\{p_h, \ket{\psi_h}\}}\sum_h p_h \sqrt{\tau(\ket{\psi_h}_{A|B})}\bigg]^2.
\label{2tangle}
\end{equation}
where the minimization is taken over all possible pure state decompositions
\begin{equation}
\rho_{AB}=\sum_{h}p_{h}\ket{\psi_h}_{AB}\bra{\psi_h}.
\label{decomp}
\end{equation}

Mathematically, monogamy of multi-party quantum entanglement was first characterized in three-qubit systems by Coffman, Kundu and Wootters(CKW)
~\cite{CKW}; using one and two tangles as the bipartite entanglement quantification, monogamy inequality of three-qubit entanglement was proposed as
\begin{equation}
\tau\left(\ket{\psi}_{A|BC}\right) \geq \tau\left(\rho_{A|B}\right)+\tau\left(\rho_{A|C}\right),
\label{eq: CKW}
\end{equation}
where  $\tau\left(\ket{\psi}_{A|BC}\right)$ is the one tangle of the three-qubit pure state $\ket{\psi}_{ABC}$ quantifying the bipartite entanglement between $A$ and $BC$,
and $\tau\left(\rho_{A|B}\right)$ and $\tau\left(\rho_{A|C}\right)$ are the two tangles
of the two-qubit reduced states $\rho_{AB}=\T_{C}\ket{\psi}_{ABC}\bra{\psi}$ and $\rho_{AC}=\T_{B}\ket{\psi}_{ABC}\bra{\psi}$,
respectively.

Later CKW inequality in (\ref{eq: CKW}) was generalized into $n$-qubit systems~\cite{OV} as
\begin{equation}
\tau\left(\ket{\psi}_{A_1|A_2\cdots A_n}\right) \geq \sum_{j=2}^{n}\tau\left(\rho_{A_1|A_j}\right),
\label{eq: OV}
\end{equation}
for one tangle $\tau\left(\ket{\psi}_{A_1|A_2\cdots A_n}\right)$ and two tangles $\tau\left(\rho_{A_1|A_j}\right)$ of
each two-qubit reduced density matrices $\rho_{A_1A_j}$ on subsystems $A_1A_j$ for each $j=2,\cdots ,n$.
However, tangle is known to fail in the generalization of CKW inequality for higher dimensional quantum systems
rather than qubits; there exist quantum states in $3 \otimes 3\otimes 3$ and even in $3 \otimes 2\otimes 2$ quantum systems violating
CKW inequality in (\ref{eq: CKW})~\cite{OU, KS}.

Now we consider another generalization of tangles from qubits to qudit systems using negativity~\cite{KDS}.
We first note that for any pure state $\ket{\psi}_{AB}$ with Schmidt-rank 2(especially for two-qubit pure state)
\begin{equation}
\ket{\psi}_{AB}=\sqrt{\lambda_1}\ket{e_0}_A\otimes\ket{f_0}_B+\sqrt{\lambda_2}\ket{e_1}_A\otimes\ket{f_1}_B,
\label{schmidt2}
\end{equation}
the {\em square} of negativity in Eq.~(\ref{eq:pure_negativity2}) coincides with the tangle in Eq.~(\ref{1tangle})
\begin{equation}
\mathcal{N}^2\left(\ket{\psi}_{A|B}\right)=4\lambda_1\lambda_2=\tau\left(\ket{\psi}_{A|B}\right).
\label{NSCeqt}
\end{equation}
Thus the two-tangle of any two-qubit state $\rho_{AB}$ in Eq.~(\ref{2tangle}) can be rephrased as
\begin{align}
\tau\left(\rho_{A|B}\right)
=&\bigg[\min_{\{p_h, \ket{\psi_h}\}}\sum_h p_h \mathcal{N}\left(\ket{\psi_h}_{A|B}\right)\bigg]^2
\label{tauscre1}
\end{align}
where the right-hand side of Eq.~(\ref{tauscre1}) is square of the convex-roof extended negativity in Eq.~(\ref{eq:c_negativity}).
Based on this idea, we propose a bipartite entanglement measure using negativity; for any two-qudit mixed state $\rho_{AB}$
its {\em square of convex-roof extended negativity}(SCREN) is defined as
\begin{align}
\mathcal{N}_{sc}(\rho_{A|B})=&\bigg[\min_{\{p_h, \ket{\psi_h}\}}\sum_h p_h \mathcal{N}\left(\ket{\psi_h}_{A|B}\right)\bigg]^2.\nonumber\\
\label{SCREN}
\end{align}

From the properties of convex-roof extended negativity in Eq.~(\ref{eq:c_negativity}), it is straightforward to check that
SCREN has the separability criterion and monotonicity under LOCC~\cite{LCOK}.
We also note that Eqs.~(\ref{tauscre1}) and (\ref{SCREN}) imply the coincidence of SCREN with two-tangle for any two-qubit state $\rho_{AB}$,
\begin{equation}
\mathcal{N}_{sc}(\rho_{A|B})=\tau\left(\rho_{A|B}\right)
\label{CRENTaneq}
\end{equation}
Consequently, the multi-qubit monogamy inequality in terms of tangles in (\ref{eq: OV}) can be
rephrased in terms of SCREN as,
\begin{equation}
{\mathcal{N}_{sc}}\left(\ket{\psi}_{A_1|A_2 \cdots A_n}\right)  \geq
\sum_{j=2}^{n}{\mathcal{N}_{sc}}\left(\rho_{A_1 |A_j}\right).
\label{nineq cren}
\end{equation}
Moreover, Inequality~(\ref{nineq cren}) still holds for the counterexamples~\cite{OU, KS} that violate
CKW inequality in higher-dimensional systems~\cite{KDS}.
Thus SCREN is a good generalization of two-qubit tangle into higher-dimensional quantum systems
without any known counterexamples even in higher-dimensional quantum systems so far.

\section{Strong Monogamy of Multi-Party Quantum Entanglement}
\label{Sec: SM_qubit}
\subsection{Multi-Qubit Strong Monogamy Inequality}
\label{sub: qubitSM}
For any three-qubit pure state $\ket{\psi}_{ABC}$, the residual entanglement from the difference between left
and right-hand sides of CKW Inequality~(\ref{eq: CKW})
is also interpreted as the genuine three-party entanglement, namely {\em three-tangle} of $\ket{\psi}_{ABC}$
\begin{equation}
\tau\left(\ket{\psi}_{A|B|C}\right)=\tau\left(\ket{\psi}_{A|BC}\right)-\tau\left(\rho_{A|B}\right)-\tau\left(\rho_{A|C}\right).
\label{3tangle}
\end{equation}
The three-tangle in Eq.~(\ref{3tangle}) is a good measure of genuine three-qubit entanglement,
which is invariant under the permutation of subsystems $A$, $B$ and $C$~\cite{DVC}.

The definition of three-tangle was generalized for arbitrary $n$-qubit quantum states~\cite{LA};
for an $n$-qubit pure state $\ket{\psi}_{A_1A_2\cdots A_n}$,
its {\em $n$-tangle} is defined as
\begin{align}
\tau\left(\ket{\psi}_{A_1|A_2|\cdots |A_n}\right)
=&\tau\left(\ket{\psi}_{A_1|A_2\cdots A_n}\right)\nonumber\\
&-\sum_{m=2}^{n-1} \sum_{\vec{j}^m}\tau\left(\rho_{A_1|A_{j^m_1}|\cdots |A_{j^m_{m-1}}}\right)^{m/2},
\label{eq:ntanglepure}
\end{align}
where the index vector $\vec{j}^m=(j^m_1,\ldots,j^m_{m-1})$ spans all the ordered subsets of the index set $\{2,\ldots,n\}$ with $(m-1)$ distinct elements.
Eq.~(\ref{eq:ntanglepure}) is a recurrent definition that needs all the $m$ tangles $\tau\left(\rho_{A_1|A_{j^m_1}|\cdots |A_{j^m_{m-1}}}\right)$
of $m$-qubit reduced density matrices $\rho_{A_1A_{j^m_1}\cdots A_{j^m_{m-1}}}$ for $2 \leq m \leq n-1$, where $\tau\left(\rho_{A_1|A_{j^m_1}|\cdots |A_{j^m_{m-1}}}\right)$
is defined as
\begin{widetext}
\begin{equation}
\tau\left(\rho_{A_1|A_{j^m_1}|\cdots |A_{j^m_{m-1}}}\right)=\bigg[\min_{\{p_h, \ket{\psi_h}\}}\sum_h p_h
\sqrt{\tau\left(\ket{\psi_h}_{A_1|A_{j^m_1}|\cdots |A_{j^m_{m-1}}}\right)}\bigg]^2,
\label{ntanglemix}
\end{equation}
\end{widetext}
with the minimization over all possible pure state decompositions
\begin{equation}
\rho_{A_1A_{j^m_1}\cdots A_{j^m_{m-1}}}=\sum_{h}p_{h}\ket{\psi_h}_{A_1A_{j^m_1}\cdots A_{j^m_{m-1}}}\bra{\psi_h}.
\label{decomp}
\end{equation}

For $n=3$, the definition of $n$-tangle in Eq.~(\ref{eq:ntanglepure}) reduces to that of three-tangle
in Eq.~(\ref{3tangle}) whose nonnegativity is equivalent to the CKW inequality (\ref{eq: CKW}).
In other words, the nonnegativity of three-tangle provides us
with a quantitative characterization of three-qubit monogamy of entanglement.
For $n=2$, Eq.~(\ref{ntanglemix}) also reduces to the two-tangle of two-qubit state $\rho_{A_1A_2}$ in Eq.~(\ref{2tangle}).

Based on this idea, a {\em strong monogamy}(SM) inequality of multi-qubit entanglement was proposed as
\begin{align}
\tau\left(\ket{\psi}_{A_1|A_2\cdots A_n}\right)\geq\sum_{m=2}^{n-1} \sum_{\vec{j}^m}\tau\left(\rho_{A_1|A_{j^m_1}|\cdots |A_{j^m_{m-1}}}\right)^{m/2}
\label{eq:SM}
\end{align}
by conjecturing the nonnegativity of $n$-tangle in Eq.~(\ref{eq:ntanglepure}).
The lower term in Inequality~(\ref{eq:SM}) appears in between the both sides of the $n$-qubit CKW inequality in (\ref{eq: OV}) as
\begin{align}
\tau\left(\ket{\psi}_{A_1|A_2\cdots A_n}\right)\geq&\sum_{j=2}^{n}\tau\left(\rho_{A_1|A_j}\right)\nonumber\\
&+\sum_{m=3}^{n-1} \sum_{\vec{j}^m}\tau\left(\rho_{A_1|A_{j^m_1}|\cdots |A_{j^m_{m-1}}}\right)^{m/2}\nonumber\\
\geq &\sum_{j=2}^{n}\tau\left(\rho_{A_1|A_j}\right),
\label{compar}
\end{align}
therefore it is a {\em stronger} inequality.
We also note that Inequality~(\ref{eq:SM}) encapsulates three-qubit CKW inequality in (\ref{eq: CKW}) for $n=3$.
Thus Inequality~(\ref{eq:SM}) is another generalization of three-qubit monogamy inequality
into multi-qubit systems in a stronger form.

For the validity of SM inequality in (\ref{eq:SM}), an extensive numerical evidence was presented
for four qubit systems together with analytical proof for some cases of multi-qubit systems.
It was also recently shown that Inequality (\ref{eq:SM}) is also true for a large class of multi-qubit generalized W-class states,
\begin{align}
\ket{\psi}_{A_1 A_2 ... A_n}=&
a_1 \ket{10\cdots0}+a_2 \ket{01\cdots0}\nonumber\\
&+...+a_n \ket{00\cdots1}
\label{supWV}
\end{align}
with
$\sum_{i=1}^{n}|a_j|^2 =1$
~\cite{Kim14}.

\subsection{SCREN Strong Monogamy Inequality}
\label{Subsec: SCRENmono}

Although Inequality~(\ref{eq:SM}) proposes a stronger monogamous property of
multi-qubit entanglement with various cases of analytic proof, Inequality~(\ref{eq:SM}) is no longer
valid for higher-dimensional quantum systems rather than qubits; for $n=3$, Inequality~(\ref{eq:SM}) becomes a CKW-type inequality of
three-party quantum systems,
\begin{equation}
\tau\left(\ket{\psi}_{A|BC}\right) \geq \tau\left(\rho_{A|B}\right)+\tau\left(\rho_{A|C}\right).
\label{eq: 3CKWtau}
\end{equation}
However, it is also known that there exists a pure state in $3 \otimes 2 \otimes 2$ quantum systems~\cite{KDS, KS},
\begin{equation}
\ket{\psi}_{ABC} = \frac{1}{\sqrt{6}}(\sqrt{2}\ket{010}+\sqrt{2}\ket{101}+\ket{200}+\ket{211}),
\label{count2}
\end{equation}
where $\tau\left(\ket{\psi}_{A|BC}\right)=\frac{12}{9}$ and $\tau\left(\rho_{A|B}\right)=\tau\left(\rho_{A|C}\right)=\frac{8}{9}$,
therefore
\begin{equation}
\tau\left(\ket{\psi}_{A|BC}\right) < \tau\left(\rho_{A|B}\right)+\tau\left(\rho_{A|C}\right).
\label{eq: counterineq}
\end{equation}
In other words, the counterexample for three-party CKW inequality in Eq.~(\ref{count2}) is also a counterexample
for SM inequality in (\ref{eq:SM}) in higher-dimensional quantum systems rather than qubits. Thus tangle-based SM inequality can
only be valid for multi-qubit systems and even a tiny extension in any of the subsystems leads to a violation.

Here we propose another generalization of multi-qubit SM inequality into higher-dimensional quantum systems using SCREN.
Due to the coincidence of tangle and SCREN for two-qubit states and any pure state of Schmidt-rank two in
Eq.~(\ref{CRENTaneq}), the definition of three-tangle in Eq.~(\ref{3tangle}) can be naturally rephrased in terms of SCREN;
for any three-qubit pure state $\ket{\psi}_{ABC}$,
\begin{align}
{\mathcal{N}_{sc}}\left(\ket{\psi}_{A|B|C}\right)=&{\mathcal{N}_{sc}}\left(\ket{\psi}_{A|BC}\right)\nonumber\\
&-{\mathcal{N}_{sc}}\left(\rho_{A|B}\right)-{\mathcal{N}_{sc}}\left(\rho_{A|C}\right).
\label{3CREN}
\end{align}
For analogous terminology, we denote ${\mathcal{N}_{sc}}\left(\ket{\psi}_{A|B|C}\right)$ in Eq.~(\ref{3CREN}) as
{\em three-SCREN} where ${\mathcal{N}_{sc}}\left(\ket{\psi}_{A|BC}\right)$ and ${\mathcal{N}_{sc}}\left(\rho_{A|B}\right)$ are {\em one-} and
{\em two-SCREN}, respectively.

Now we generalize the definition of three-SCREN in Eq.~(\ref{3CREN}) into arbitrary multi-party, higher-dimensional quantum systems.
For an $n$-qudit pure state $\ket{\psi}_{A_1A_2\cdots A_n}$,
its {\em $n$-SCREN} is defined as
\begin{widetext}
\begin{align}
{\mathcal{N}_{sc}}\left(\ket{\psi}_{A_1|A_2|\cdots |A_n}\right)
={\mathcal{N}_{sc}}\left(\ket{\psi}_{A_1|A_2\cdots A_n}\right)
-\sum_{m=2}^{n-1} \sum_{\vec{j}^m}{\mathcal{N}_{sc}}\left(\rho_{A_1|A_{j^m_1}|\cdots |A_{j^m_{m-1}}}\right)^{m/2},
\label{eq:nCRENpure}
\end{align}
\end{widetext}
where ${\mathcal{N}_{sc}}\left(\ket{\psi}_{A_1|A_2\cdots A_n}\right)$ is the one-SCREN of $n$-qudit pure state with respect to the bipartition between $A_1$ and the other qudit systems,
and
the $m$-SCREN
of $m$-qubit reduced density matrix $\rho_{A_1A_{j^m_1}\cdots A_{j^m_{m-1}}}$
is defined as
\begin{widetext}
\begin{equation}
{\mathcal{N}_{sc}}\left(\rho_{A_1|A_{j^m_1}|\cdots |A_{j^m_{m-1}}}\right)=\bigg[\min_{\{p_h, \ket{\psi_h}\}}\sum_h p_h
\sqrt{{\mathcal{N}_{sc}}\left(\ket{\psi_h}_{A_1|A_{j^m_1}|\cdots |A_{j^m_{m-1}}}\right)}\bigg]^2,
\label{nCRENmix}
\end{equation}
\end{widetext}
with the minimization over all possible pure state decompositions of $\rho_{A_1A_{j^m_1}\cdots A_{j^m_{m-1}}}$.
We also note that the index vector $\vec{j}^m=(j^m_1,\ldots,j^m_{m-1})$ in the second summation of
Eq.~(\ref{eq:nCRENpure}) spans all the ordered subsets of the index set $\{2,\ldots,n\}$ with $(m-1)$ distinct elements.

For a multi-qudit pure state $\ket{\psi}_{A_1A_2\cdots A_n}$, the {\em SCREN-SM inequality} of multi-party entanglement can be derived as
\begin{align}
{\mathcal{N}_{sc}}\left(\ket{\psi}_{A_1|A_2\cdots A_n}\right)\geq\sum_{m=2}^{n-1} \sum_{\vec{j}^m}{\mathcal{N}_{sc}}\left(\rho_{A_1|A_{j^m_1}|\cdots |A_{j^m_{m-1}}}\right)^{m/2},
\label{eq:CRENSM}
\end{align}
conjecturing the nonnegativity of $n$-SCREN in Eq.~(\ref{eq:nCRENpure}).
From the relation of SCREN and tangle in Eq.~(\ref{CRENTaneq}), Inequality~(\ref{eq:CRENSM}) is reduced to
Inequality~(\ref{eq:SM}) for any multi-qubit states. Thus Inequality~(\ref{eq:CRENSM}) is a generalization
of multi-qubit SM inequality in terms of tangle, which is valid for the classes of multi-qubit quantum states considered in~\cite{LA, Kim14}.

For the counterexample of CKW inequality in Eq.~(\ref{count2}), it is straightforward to check ${\mathcal{N}_{sc}}\left(\ket{\psi}_{A|BC}\right)=4$ whereas
${\mathcal{N}_{sc}}\left(\rho_{A|B}\right) = {\mathcal{N}_{sc}}\left(\rho_{A|C}\right)=\frac{8}{9}$, and thus
\begin{equation}
{\mathcal{N}_{sc}}\left(\ket{\psi}_{A|BC}\right) \geq {\mathcal{N}_{sc}}\left(\rho_{A|B}\right)+{\mathcal{N}_{sc}}\left(\rho_{A|C}\right).
\label{eq: 322CKWCREN}
\end{equation}
Moreover, for the other counterexample in $3\otimes 3\otimes 3$ quantum systems~\cite{OU},
\begin{align}
|\psi\rangle_{ABC}=\frac{1}{\sqrt{6}}(&|123\rangle-|132\rangle+|231\rangle\nonumber\\
&-|213\rangle+|312\rangle-|321\rangle),
\label{f}
\end{align}
we have ${\mathcal{N}_{sc}}\left(\ket{\psi}_{A|BC}\right)=4$ whereas
${\mathcal{N}_{sc}}\left(\rho_{A|B}\right) = {\mathcal{N}_{sc}}\left(\rho_{A|C}\right)=1$.
In other words, Inequality~(\ref{eq: 322CKWCREN}) is still true for all the known counterexamples of CKW inequality, therefore
SCREN is a good alternative of tangle in characterizing strongly monogamous property of multi-party entanglement.

\section{SCREN Strong Monogamy Inequality of Multi-qudit Entanglement}
\label{sec: SMgW}
\subsection{Multi-Qudit Generalized W-class States}
\label{Subsec: gW}
Let us recall the definition of multi-qudit generalized W-class
state~\cite{KS},
\begin{align}
\left|W_n^d \right\rangle_{A_1\cdots A_n}=\sum_{i=1}^{d-1}(&a_{1i}{\ket {i0\cdots 0}} +a_{2i}{\ket {0i\cdots 0}}\nonumber\\
&+\cdots +a_{ni}{\ket {00\cdots 0i}}),
\label{generalized W state}
\end{align}
with the normalization condition $\sum_{s=1}^{n}\sum_{i=1}^{d-1}|a_{si}|^2=1$.
The state in Eq.~(\ref{generalized W state}) is a coherent superposition of all $n$-qudit
product states with Hamming weight one. We also note that the term ``{\em generalized}'' naturally arises
because Eq.~(\ref{generalized W state}) includes $n$-qubit W-class states in Eq.~(\ref{supWV}) as a special case
when $d=2$.

Before we further investigate strongly monogamous property of entanglement for this generalized W-class state,
we first recall a very useful property of quantum states proposed by Hughston-Jozsa-Wootters(HJW) showing the unitary freedom
in the ensemble for density matrices~\cite{HJW}.
\begin{Prop} (HJW theorem)
The sets $\{|\tilde{\phi_i}\rangle\}$ and $\{|\tilde{\psi_j}\rangle\}$ of (possibly unnormalized) states generate the same density matrix
if and only if
\begin{equation}
|\tilde{\phi_i}\rangle=\sum_j u_{ij}|\tilde{\psi_j}\rangle\
\label{HJWeq}
\end{equation}
where $(u_{ij})$ is a unitary matrix of complex numbers, with indices $i$ and $j$, and we
{\em pad} whichever set of states $\{|\tilde{\phi_i}\rangle\}$ or $\{|\tilde{\psi_j}\rangle\}$ is smaller with additional zero vectors
so that the two sets have the same number of elements.
\label{HJWthm}
\end{Prop}

A direct consequence of Proposition~\ref{HJWthm} is the following; for two pure-state decompositions
$\sum_{i}p_{i}\ket{\phi_i}\bra{\phi_i}$ and $\sum_{j}q_{j}\ket{\psi_j}\bra{\psi_j}$,
they represent the same density matrix, that is $\rho=\sum_{i}p_{i}\ket{\phi_i}\bra{\phi_i}=\sum_{j}q_{j}\ket{\psi_j}\bra{\psi_j}$
if and only if $\sqrt{p_{i}}\ket{\phi_i}=\sum_{j}u_{ij}\sqrt{q_{j}}\ket{\psi_j}$ for some unitary matrix $u_{ij}$.
Using Proposition~\ref{HJWthm}, we provide the following lemma, which shows a structural property of multi-qudit generalized W-class states.
\begin{Lem}
Let $\ket{\psi}_{A_1\cdots A_n}$ be a $n$-qudit pure state in a superposition of a $n$-qudit generalized W-class state
in Eq.~(\ref{supWV}) and vacuum, that is,
\begin{equation}
\ket{\psi}_{A_1A_2\cdots A_n}=\sqrt{p}\left|W_n^d \right\rangle_{A_1\cdots A_n}+\sqrt{1-p}\ket{0\cdots 0}_{A_1\cdots A_n}
\label{WV2}
\end{equation}
for $0\leq p \leq 1$. Let $\rho_{A_1A_{j_1}\cdots A_{j_{m-1}}}$ be a reduced density matrix of $\ket{\psi}_{A_1\cdots A_n}$
onto $m$-qudit subsystems $A_1A_{j_1}\cdots A_{j_{m-1}}$ with $2 \leq m \leq  n-1$.
For any pure state decomposition of $\rho_{A_1A_{j_1}\cdots A_{j_{m-1}}}$ such that
\begin{align}
\rho_{A_1A_{j_1}\cdots A_{j_{m-1}}}=\sum_{k}q_k\ket{\phi_k}_{A_1A_{j_1}\cdots A_{j_{m-1}}}\bra{\phi_k},
\label{rhoa1aj1ajm-1}
\end{align}
$\ket{\phi_k}_{A_1A_{j_1}\cdots A_{j_{m-1}}}$ is a superposition of a $m$-qudit generalized W-class state and vacuum.
\label{Lem: reduced}
\end{Lem}

\begin{proof}
Due to the symmetry of the structure of multi-qudit generalized W-class states with respect to permuting subsystems,
here we only consider the reduced density matrix $\rho_{A_1A_2\cdots A_m}$ of the first $m$ qudits subsystems $A_1A_2\cdots A_m$, where the general
cases of $m$-qudit subsystems $A_1A_{j_1}\cdots A_{j_{m-1}}$ is then analogously following.

From a straightforward calculation, we obtain
\begin{align}
\rho_{A_1A_2\cdots A_m}=\ket{\tilde{x}}_{A_1A_2\cdots A_m}\bra{\tilde{x}}
+\ket{\tilde{y}}_{A_1A_2\cdots A_m}\bra{\tilde{y}},
\label{mrho}
\end{align}
where
\begin{widetext}
\begin{align}
\ket{\tilde{x}}_{A_1A_2\cdots A_m}=
&\sqrt{p}\sum_{i=1}^{d-1}\left(a_{1i}\ket{i0\cdots0}_{A_1A_2\cdots A_m}+a_{2i}\ket{0i0\cdots0}_{A_1A_2\cdots A_m}+
\cdots+a_{mi}\ket{00\cdots i}_{A_1A_2\cdots A_m}\right)\nonumber\\
&+\sqrt{1-p}\ket{00\cdots 0}_{A_1A_2\cdots A_m},\nonumber\\
\ket{\tilde{y}}_{A_1A_2\cdots A_m}=&\sqrt{p\sum_{i=1}^{d-1}\left(|a_{m+1i}|^2+\cdots +|a_{ni}|^2\right)}\ket{00\cdots 0}_{A_1A_2\cdots A_m}
\label{xym}
\end{align}
\end{widetext}
are the unnormalized states in $m$-qubit subsystems $A_1A_2\cdots A_m$.

Now, let us consider the unnormalized states $\tilde{\ket{\phi_k}}_{A_1A_2\cdots A_m}=\sqrt{q_k}\ket{\phi_k}_{A_1A_2\cdots A_m}$ for each $k$
in the pure-state decomposition Eq.~(\ref{rhoa1aj1ajm-1}). From Proposition~\ref{HJWthm},
there exists an $r\times r$ unitary matrix $(u_{kl})$ such that
\begin{equation}
|\tilde{\phi_k}\rangle_{A_1A_2\cdots A_m}=u_{k1}\ket{\tilde{x}}_{A_1A_2\cdots A_m}+u_{k2}\ket{\tilde{y}}_{A_1A_2\cdots A_m},
\label{HJWrelation2}
\end{equation}
for each $k$. Moreover, Eqs.~(\ref{xym}) imply that both $\ket{\tilde{x}}_{A_1A_2\cdots A_m}$ and $\ket{\tilde{y}}_{A_1A_2\cdots A_m}$ are
linear combinations of $m$-qudit generalized W-class states and vacuums. In other words, $|\tilde{\phi_k}\rangle_{A_1A_2\cdots A_m}$ in Eq.~(\ref{HJWrelation2})
is an unnormalized superposition of a $m$-qudit generalized W-class state and vacuum for each $k$. Thus the same is true for the normalized state $\ket{\phi_k}_{A_1A_2\cdots A_m}$
for each $k$.
\end{proof}

\subsection{SCREN Strong Monogamy Inequality and Generalized W-class States}
\label{subsec: Sat}

In this scetion, we prove that the multi-qudit SCREN SM inequality of entanglement is true for a large class of multi-qubit quantum states in Eq.~(\ref{WV2});
superposition of multi-qudit generalized W-class states and vacuums. We first provide the following theorem about
the multi-qudit generalized W-class and the CKW-type monogamy inequality.

\begin{Thm}
For a $n$-qudit pure state
\begin{equation}
\ket{\psi}_{A_1A_2\cdots A_n}=\sqrt{p}\left|W_n^d \right\rangle_{A_1\cdots A_n}+\sqrt{1-p}\ket{0\cdots 0}_{A_1\cdots A_n}
\label{WV3}
\end{equation}
where $\left|W_n^d \right\rangle_{A_1\cdots A_n}$ is a $n$-qudit generalized W-class state in Eq.~(\ref{supWV}) and
$\ket{0\cdots 0}_{A_1\cdots A_n}$ is the vacuum, we have
\begin{align}
{\mathcal{N}_{sc}}\left(\ket{\psi}_{A_1|A_2\cdots A_n}\right)=
{\mathcal{N}_{sc}}\left(\rho_{A_1 |A_2}\right) +\cdots+{\mathcal{N}_{sc}}\left(\rho_{A_1|A_n}\right),
\label{eq: SCsat}
\end{align}
where ${\mathcal{N}_{sc}}\left(\ket{\psi}_{A_1|A_2\cdots A_n}\right)$ is the ons-SCREN of $\ket{\psi}_{A_1A_2\cdots A_n}$ with respect to the bipartition
between $A_1$ and the other qudits, and ${\mathcal{N}_{sc}}\left(\rho_{A_1 |A_s}\right)$ is the two-SCREN of of the two-qudit state $\rho_{A_1A_s}$
with $s=2,\cdots ,n$.
\label{SCsat}
\end{Thm}

\begin{proof}
For the one-SCREN of $\ket{\psi}_{A_1 \cdots A_n}$ with respect to the bipartition between $A_1$ and the other qudits,
the reduced density matrix $\rho_{A_1}$ of $\ket{\psi}_{A_1 \cdots A_n}$ onto subsystem $A_1$ is obtained as
\begin{align}
\rho_{A_1}=&\T_{A_2\cdots A_n}\ket{\psi}_{A_1 A_2 ... A_n}\bra{\psi}\nonumber\\
=&p\sum_{i,j=1}^{d-1}a_{1i}a^{*}_{1j}\ket{i}_{A_1}\bra{j}+\left[p\Omega+\left(1-p\right)\right]\ket{0}_{A_1}\bra{0}\nonumber\\
&+\sqrt{p\left(1-p\right)}\left[\sum_{i=1}^{d-1}a_{1i}\ket{i}_{A_1}\bra{0}
+\sum_{j=1}^{d-1}a^{*}_{1j}\ket{0}_{A_1}\bra{j}\right],
\label{rho_A_1}
\end{align}
where $\Omega=\sum_{s=2}^{n}\sum_{i=1}^{d-1}|a_{si}|^2=1-\sum_{j=1}^{d-1}|a_{1j}|^2$.

From the the definition of pure state negativity in Eq.~(\ref{eq:pure_negativity2}) together with Eq.~(\ref{rho_A_1}), we have the one-SCREN of
$\ket{\psi}_{A_1A_2\cdots A_n}$ between $A_1$ and the other qudits as
\begin{align}
{\mathcal{N}_{sc}}\left(\ket{\psi}_{A_1|A_2\cdots A_n}\right)=&\left((\T{ \sqrt{\rho_{A_1}}})^2 -1\right)^2\nonumber\\
=&4p^2\left(1-\Omega\right)\Omega.
\label{1SCWpure}
\end{align}

For the two-SCREN's ${\mathcal{N}_{sc}}\left(\rho_{A_1A_s}\right)$ with $s=2,\cdots , n$ that appear the right-hand side of Eq.~(\ref{eq: SCsat}),
we first consider the case when $s=2$, where all the other cases are analogously following.
The two-qudit reduced density matrix $\rho_{A_1A_2}$ of $\ket{\psi}_{A_1 A_2 ... A_n}$ is obtained as
\begin{widetext}
\begin{align}
\rho_{A_1A_2}=&\T_{A_3\cdots A_n}\ket{\psi}_{A_1 A_2 ... A_n}\bra{\psi}\nonumber\\
=&p\sum_{i,j=1}^{d-1}\left[a_{1i}a^*_{1j}\ket{i0}_{A_1A_2}\bra{j0}+a_{1i}a^*_{2j}\ket{i0}_{A_1A_2}\bra{0j}+a_{2i}a^*_{1j}\ket{0i}_{A_1A_2}\bra{j0}
+a_{2i}a^*_{2j}\ket{0i}_{A_1A_2}\bra {0j}\right]\nonumber\\
&+(\Omega_{2}+1-p)\ket{00}_{A_1A_2}\bra{00}\nonumber\\
&+\sqrt{p(1-p)}\sum_{k=1}^{d-1}\left[(a_{1k}\ket{k0}+a_{2k}\ket{0k})_{A_1A_2}\bra{00}
+a^*_{1k}\ket{00}_{A_1A_2}(\bra {k0}+a^*_{2k}\bra {0k})\right],
\label{rho12 matrix}
\end{align}
\end{widetext}
with $\Omega_{2}=1-\sum_{j=1}^{d-1}(|a_{1j}|^2+|a_{2j}|^2)$.
We further note that, by considering two unnormalized states
\begin{align}
\ket{\tilde{x}}_{A_1 A_2}=&\sqrt{p}\sum_{i=1}^{d-1}(a_{1i}\ket {i0}_{A_1 A_2}+a_{2i}\ket
{0i}_{A_1 A_2})\nonumber\\
&~~~~~~~~+\sqrt{1-p}\ket {00}_{A_1 A_2},\nonumber\\
\ket{\tilde{y}}=&\sqrt{\Omega_{2}}\ket {00}_{A_1 A_2},
\end{align}
$\rho_{A_1A_2}$ in Eq.~(\ref{rho12 matrix}) can be represented as
\begin{equation}
\rho_{A_1A_2}=\ket{\tilde{x}}_{A_1 A_2}\bra{\tilde{x}}+\ket{\tilde{y}}_{A_1 A_2}\bra{\tilde{y}}.
\end{equation}

Now Proposition~\ref{HJWthm} implies that
for any pure state decomposition
\begin{equation}
\rho_{A_1 A_2}=\sum_{h}|\tilde{\phi_h}\rangle_{A_1 A_2} \langle\tilde{\phi_h}|,
\label{decomp}
\end{equation}
where
$|\tilde{\phi_h}\rangle_{A_1 A_2}$ is an unnormalized state in two-qudit subsystem $A_1A_2$,
there exists an $r\times r$ unitary matrix $(u_{hl})$ such that
\begin{equation}
|\tilde{\phi_h}\rangle_{A_1A_2}=u_{h1}\ket{\tilde{x}}_{A_1 A_2}+u_{h2}\ket{\tilde{y}}_{A_1 A_2},
\label{HJWrelation}
\end{equation}
for each $h$.
For the normalized state $\ket{\phi_h}_{A_1 A_2}=|\tilde{\phi}_h\rangle_{A_1 A_2}/\sqrt{p_h}$
with $ p_h =|\langle\tilde{\phi}_h|\tilde{\phi}_h\rangle|$, the definition of pure state negativity in Eq.~(\ref{eq:pure_negativity2})
leads us to the two-SCREN of $\ket{\phi_h}_{A_1 A_2}$,
\begin{align}
{\mathcal{N}_{sc}}\left(\ket{\phi_h}_{A_1|A_2}\right)=&\frac{4}{p_h^2}p^2|u_{h2}|^4\left(1-\Omega\right)\left(\Omega-\Omega_{2}\right)\nonumber\\
=&\frac{4}{p_h^2}p^2|u_{h2}|^4\left(1-\Omega\right)\sum_{i=1}^{d-1}|a_{2i}|^2
\label{SCRENphi_h}
\end{align}
for each $h$.

From the definition of SCREN for mixed states in Eq.~(\ref{nCRENmix}) together with Eq.~(\ref{SCRENphi_h}),
we have the two-SCREN of $\rho_{A_1A_2}$ as
\begin{align}
{\mathcal{N}_{sc}}\left(\rho_{A_1|A_2}\right)=&\bigg[\min_{\{p_h, \ket{\phi_h}\}}\sum_h p_h \sqrt{{\mathcal{N}_{sc}}\left(\ket{\phi_h}_{A_1|A_2}\right)}\bigg]^2\nonumber\\
=&\left[\min_{\{p_h, \ket{\phi_h}\}}\sum_h 2p|u_{h2}|^2\sqrt{\left(1-\Omega\right)\sum_{i=1}^{d-1}|a_{2i}|^2}\right]^2\nonumber\\
=&4p^2\left(1-\Omega\right)\sum_{i=1}^{d-1}|a_{2i}|^2.
\label{2SCREN12}
\end{align}
where the last equality is due to the choice of $u_{h2}$ from the unitary matrix $(u_{hl})$.
Here we note that the minimum average of the square-root of SCREN in Eq.~(\ref{2SCREN12}) does not depend on the choice of pure-state decomposition
of $\rho_{A_1A_2}$, so that we could circumvent the minimization problem therein.

By using an analogous method, we have the two-SCREN of two-qudit mixed state $\rho_{A_1A_s}$ as
\begin{align}
{\mathcal{N}_{sc}}\left(\rho_{A_1|A_s}\right)=&4p^2\left(1-\Omega\right)\sum_{i=1}^{d-1}|a_{si}|^2,
\label{2SCRENs}
\end{align}
for each $s=2,\cdots ,n$.
Now Eqs.~(\ref{1SCWpure}) and (\ref{2SCRENs}) leads us to
\begin{align}
{\mathcal{N}_{sc}}\left(\ket{\psi}_{A_1|A_2\cdots A_n}\right)=&4p^2\left(1-\Omega\right)\Omega\nonumber\\
=&4p^2\left(1-\Omega\right)\sum_{s=2}^{n}\sum_{i=1}^{d-1}|a_{si}|^2\nonumber\\
=&\sum_{s=2}^{n}\left[4p^2\left(1-\Omega\right)\sum_{i=1}^{d-1}|a_{si}|^2\right]\nonumber\\
=&\sum_{s=2}^{n}{\mathcal{N}_{sc}}\left(\rho_{A_1|A_s}\right),
\label{sums}
\end{align}
which completes the proof.
\end{proof}
Theorem~\ref{SCsat} implies that Inequality (\ref{nineq cren}), the multi-qubit CKW inequality in terms of one, and two-SCREN,
is still true and in fact saturated for the class of multi-qudit states in Eq.~(\ref{WV2}).

To check the validity of SCREN SM inequality in (\ref{eq:CRENSM}) for the class of states
in Eq.~(\ref{WV2}), we first note that Inequality (\ref{eq:CRENSM}) can be decomposed as
\begin{align}
{\mathcal{N}_{sc}}\left(\ket{\psi}_{A_1|A_2\cdots A_n}\right)\geq&\sum_{m=3}^{n-1} \sum_{\vec{j}^m}{\mathcal{N}_{sc}}
\left(\rho_{A_1|A_{j^m_1}|\cdots |A_{j^m_{m-1}}}\right)^{m/2}\nonumber\\
&+\sum_{j=2}^{n}{\mathcal{N}_{sc}}\left(\rho_{A_1|A_j}\right),\nonumber\\
\label{compar2}
\end{align}
where the second summation of the first term on the right-hand side of the inequality runs over all
the index vectors $\vec{j}^m=(j^m_1,\ldots,j^m_{m-1})$ with $3\leq m \leq n-1$.

By Theorem~\ref{SCsat},
the last term of the right-hand side and the left-hand side of
of Inequality~(\ref{compar2}) are equal to each other for the class of states in Eq.~(\ref{WV2}). Thus
this class of states are good candidates as possible counterexamples
for stronger version of monogamy inequalities, that is, Inequality (\ref{eq:CRENSM}).
Moreover, the validity of SCREN SM inequality for this class of states necessarily implies that Inequality
(\ref{eq:CRENSM}) must be saturated, that is, the residual term
\begin{align}
\sum_{m=3}^{n-1} \sum_{\vec{j}^m}{\mathcal{N}_{sc}}\left(\rho_{A_1|A_{j^m_1}|\cdots |A_{j^m_{m-1}}}\right)^{m/2}
\label{resizero}
\end{align}
in (\ref{compar2}) is zero for the class of states in Eq.~(\ref{WV2}).
The following theorem states the main result of this paper, the saturation of multi-qudit SM inequality for the class of states in in Eq.~(\ref{WV2}).

\begin{Thm}
For the class of $n$-qudit states $\ket{\psi}_{A_1A_2 \cdots A_n}$ in Eq.~(\ref{WV2}) that is a superposition of a $n$-qudit generalized W-class state and the vacuum,
the multi-qudit SM inequality of entanglement in terms of SCREN is saturated;
\begin{align}
{\mathcal{N}_{sc}}\left(\ket{\psi}_{A_1|A_2\cdots A_n}\right)=\sum_{m=2}^{n-1} \sum_{\vec{j}^m}{\mathcal{N}_{sc}}\left(\rho_{A_1|A_{j^m_1}|\cdots |A_{j^m_{m-1}}}\right)^{m/2}.
\label{eq:SMsat}
\end{align}
\label{SMsat}
\end{Thm}
\begin{proof}
As mentioned, it is enough to show that the residual term in Eq.~(\ref{resizero}) is zero for the class of states in Eq.~(\ref{WV2}).
In fact, we further show that
\begin{align}
{\mathcal{N}_{sc}}\left(\rho_{A_1|A_{j^m_1}|\cdots |A_{j^m_{m-1}}}\right)=0
\label{mSCREN0}
\end{align}
for all the index vectors $\vec{j}^m=(j^m_1,\ldots,j^m_{m-1})$ with $3\leq m \leq n-1$, that is,
all the $m$-SCREN for $3\leq m \leq n-1$ is zero for the $m$-qudit reduced density matrices $\rho_{A_1A_{j^m_1}\cdots A_{j^m_{m-1}}}$.

We use the mathematical induction on $m$, and first consider the case when $m=3$.
For any index vector $\vec{j}^3=(j_1, j_2)$ with $j_1,~j_2 \in \{2,3,\cdots,n\}$~\cite{omit},
the left-hand side of Eq.~(\ref{mSCREN0}) becomes the three-SCREN of the three-qudit
reduced density matrix $\rho_{A_1A_{j_1}A_{j_2}}$,
\begin{align}
{\mathcal{N}_{sc}}&\left(\rho_{A_1|A_{j_1}|A_{j_2}}\right)\nonumber\\
&~~~=\bigg[\min_{\{p_h, \ket{\psi_h}\}}\sum_h p_h
\sqrt{{\mathcal{N}_{sc}}\left(\ket{\psi_h}_{A_1|A_{j_1}|A_{j_2}}\right)}\bigg]^2,
\label{3CRENmix}
\end{align}
where the minimization is over all possible pure state decompositions of $\rho_{A_1A_{j_1}A_{j_2}}$.
Let us consider an optimal decomposition
\begin{align}
\rho_{A_1A_{j_1}A_{j_2}}=\sum_{k}q_{k}\ket{\phi_k}_{A_1|A_{j_1}|A_{j_2}}\bra{\phi_k},
\label{3decoopt}
\end{align}
realizing the three-SCREN of of $\rho_{A_1A_{j_1}A_{j_2}}$,
\begin{align}
{\mathcal{N}_{sc}}&\left(\rho_{A_1|A_{j_1}|A_{j_2}}\right)\nonumber\\
&~~~~~~~=\bigg[\sum_k q_k\sqrt{{\mathcal{N}_{sc}}\left(\ket{\phi_k}_{A_1|A_{j_1}|A_{j_2}}\right)}\bigg]^2.
\label{3SCRENopt}
\end{align}

Because $\rho_{A_1A_{j_1}A_{j_2}}$ is a three-qudit reduced density matrix of $\ket{\psi}_{A_1A_2\cdots A_n}$ in Eq.~(\ref{WV2}),
Lemma~\ref{Lem: reduced} implies that $\ket{\phi_k}_{A_1A_{j_1}A_{j_2}}$ in Eq.~(\ref{3decoopt}) is a superposition of a three-qudit
generalized W-class state and vacuum for each $k$. Due to Theorem~\ref{SCsat}, we also note that CKW-type monogamy inequality
in terms of SCREN is saturated by $\ket{\phi_k}_{A_1A_{j_1}A_{j_2}}$ in Eq.~(\ref{3decoopt}) for each $k$;
\begin{align}
{\mathcal{N}_{sc}}\left(\ket{\phi_k}_{A_1|A_{j_1}A_{j_2}}\right)=
{\mathcal{N}_{sc}}\left(\rho^k_{A_1 |A_{j_1}}\right)+{\mathcal{N}_{sc}}\left(\rho^k_{A_1|A_{j_2}}\right),
\label{eq: 3SCsat}
\end{align}
where $\rho^k_{A_1A_{j_1}}$ and $\rho^k_{A_1A_{j_2}}$ are the reduced density matrices of $\ket{\phi_k}_{A_1A_{j_1}A_{j_2}}$
onto two-qudit subsystems $A_1A_{j_1}$ and $A_1A_{j_2}$ respectively.

From the definition of pure-state SCREN in Eq.~(\ref{eq:nCRENpure})
together with Eq.~(\ref{eq: 3SCsat}), we have
\begin{align}
{\mathcal{N}_{sc}}\left(\ket{\phi_k}_{A_1|A_{j_1}|A_{j_2}}\right)=&{\mathcal{N}_{sc}}\left(\ket{\phi_k}_{A_1|A_{j_1}A_{j_2}}\right)\nonumber\\
&-{\mathcal{N}_{sc}}\left(\rho^k_{A_1 |A_{j_1}}\right)-{\mathcal{N}_{sc}}\left(\rho^k_{A_1|A_{j_2}}\right)\nonumber\\
=&0,
\label{3SCREN0}
\end{align}
for each three-qudit pure state $\ket{\phi_k}_{A_1|A_{j_1}A_{j_2}}$ in Eq.~(\ref{3SCRENopt}),
and thus we have
\begin{align}
{\mathcal{N}_{sc}}&\left(\rho_{A_1|A_{j_1}|A_{j_2}}\right)=0
\label{3SCRENmix0}
\end{align}
for any three-qudit reduced density matrix $\rho_{A_1A_{j_1}A_{j_2}}$ of $\ket{\psi}_{A_1A_2\cdots A_n}$ in Eq.~(\ref{WV2}).

Now we assume the induction hypothesis; for any $(m-1)$-qudit reduced density matrix
$\rho_{A_1A_{j_1}A_{j_2}\cdots A_{j_{m-2}}}$ of the state $\ket{\psi}_{A_1A_2\cdots A_n}$ in Eq.~(\ref{WV2}), we assume
\begin{align}
{\mathcal{N}_{sc}}\left(\rho_{A_1|A_{j_1}|A_{j_2}|\cdots|A_{j_{m-2}}}\right)=0.
\label{induct}
\end{align}
For any index vector $\vec{j}=(j_1, j_2, \ldots, j_{m-1})$ with $\{j_1,~j_2, \ldots, j_{m-1}\}\subseteq\{2,3,\cdots,n\}$
and the $m$-qudit reduced density matrix $\rho_{A_1A_{j_1}\cdots A_{j_{m-1}}}$,
we consider an optimal pure-state decomposition
\begin{align}
\rho_{A_1A_{j_1}\cdots A_{j_{m-1}}}=\sum_{k}q_k\ket{\phi_k}_{A_1A_{j_1}\cdots A_{j_{m-1}}}\bra{\phi_k}
\label{optdecomp}
\end{align}
realizing $m$-SCREN of $\rho_{A_1A_{j_1}\cdots A_{j_{m-1}}}$, that is,
\begin{align}
{\mathcal{N}_{sc}}&\left(\rho_{A_1|A_{j_1}|\cdots |A_{j_{m-1}}}\right)\nonumber\\
&~~~~=\bigg[\sum_k q_k\sqrt{{\mathcal{N}_{sc}}\left(\ket{\phi_k}_{A_1|A_{j_1}|\cdots |A_{j_{m-1}}}\right)}\bigg]^2.
\label{mSCRENopt}
\end{align}

From the definition of pure-state SCREN in Eq.~(\ref{eq:nCRENpure}), the $m$-SCREN of each
$\ket{\phi_k}_{A_1A_{j_1}\cdots A_{j_{m-1}}}$ in Eq.~(\ref{mSCRENopt}) is
\begin{align}
{\mathcal{N}_{sc}}\left(\ket{\phi_k}_{A_1|A_{j_1}|\cdots|A_{j_{m-1}}}\right)
=&{\mathcal{N}_{sc}}\left(\ket{\phi_k}_{A_1|A_{j_1}\cdots A_{j_{m-1}}}\right)\nonumber\\
-\sum_{s=2}^{m-1} \sum_{\vec{i}^s}&{\mathcal{N}_{sc}}\left(\rho^k_{A_1|A_{i_1}|\cdots |A_{i_{s-1}}}\right)^{s/2},
\label{mSCRENpure1}
\end{align}
where $\rho^k_{A_1A_{i_1}\cdots A_{i_{s-1}}}$ is the reduced density matrix of
$\ket{\phi_k}_{A_1A_{j_1}\cdots A_{j_{m-1}}}$ on $s$-qudit subsystems ${A_1A_{i_1}\cdots A_{i_{s-1}}}$,
and the second summation is over all possible index vectors $\vec{i}^s=(i_1, i_2, \cdots, i_{s-1})$ with $\{i_1,~i_2, \cdots, i_{s-1}\} \subseteq \{j_1, j_2,\cdots,j_{m-1}\}$.
We further divide the last term of the right-hand side of Eq.~(\ref{mSCRENpure1}) into the summations of two-SCREN and the others;
\begin{align}
{\mathcal{N}_{sc}}\left(\ket{\phi_k}_{A_1|A_{j_1}|\cdots|A_{j_{m-1}}}\right)
=&{\mathcal{N}_{sc}}\left(\ket{\phi_k}_{A_1|A_{j_1}\cdots A_{j_{m-1}}}\right)\nonumber\\
&-\sum_{l=1}^{m-1}{\mathcal{N}_{sc}}\left(\rho^k_{A_1|A_{j_l}} \right)\nonumber\\
-\sum_{s=3}^{m-1} \sum_{\vec{i}^s}&{\mathcal{N}_{sc}}\left(\rho^k_{A_1|A_{i_1}|\cdots |A_{i_{s-1}}}\right)^{s/2}.
\label{mSCRENpure2}
\end{align}

For each $s=3, \cdots ,m-1$, $\rho^k_{A_1A_{i_1}\cdots A_{i_{s-1}}}$ in the last summation of Eq.~(\ref{mSCRENpure2})
is a $s$-qudit reduced density matrix of the $m$-qudit state $\ket{\phi_k}_{A_1A_{j_1}\cdots A_{j_{m-1}}}$ where Lemma~\ref{Lem: reduced}
implies that $\ket{\phi_k}_{A_1A_{j_1}\cdots A_{j_{m-1}}}$ in Eq.~(\ref{optdecomp}) is a superposition of a $m$-qudit W-class state and vacuum.
Thus the induction hypothesis assures that the $s$-SCREN of $\rho^k_{A_1A_{i_1}\cdots A_{i_{s-1}}}$ is zero;
\begin{equation}
{\mathcal{N}_{sc}}\left(\rho^k_{A_1|A_{i_1}|\cdots |A_{i_{s-1}}}\right)=0,
\label{sSCREN0}
\end{equation}
for each $s=3,\cdots ,m-1$ and the index vector $\vec{i}^s=(i_1, i_2, \cdots, i_{s-1})$.

Furthermore, Theorem~\ref{SCsat} implies that the CKW-type monogamy inequality in terms of one and two SCREN is saturated by
$\ket{\phi_k}_{A_1A_{j_1}\cdots A_{j_{m-1}}}$, that is,
\begin{align}
{\mathcal{N}_{sc}}\left(\ket{\phi_k}_{A_1|A_{j_1}\cdots A_{j_{m-1}}}\right)=\sum_{l=1}^{m-1}{\mathcal{N}_{sc}}\left(\rho^k_{A_1|A_{j_l}} \right),
\label{satphi}
\end{align}
for each $k$. From Eq.~(\ref{mSCRENpure2}) together with Eqs.~(\ref{sSCREN0}) and (\ref{satphi}), we have
\begin{align}
{\mathcal{N}_{sc}}\left(\ket{\phi_k}_{A_1|A_{j_1}|\cdots|A_{j_{m-1}}}\right)=0
\label{mSCRENphi0}
\end{align}
for each $\ket{\phi_k}_{A_1A_{j_1}\cdots A_{j_{m-1}}}$ that arises in the decomposition of $\rho_{A_1A_{j_1}\cdots A_{j_{m-1}}}$ in Eq.~(\ref{optdecomp}).
Thus Eqs.~(\ref{mSCRENopt}) and (\ref{mSCRENphi0}) lead us to
\begin{align}
{\mathcal{N}_{sc}}\left(\rho_{A_1|A_{j_1}|\cdots|A_{j_{m-1}}} \right)=0,
\label{SCRENrhom0}
\end{align}
for any the $m$-qudit reduced density matrix $\rho_{A_1A_{j_1}\cdots A_{j_{m-1}}}$ of $\ket{\psi}_{A_1 A_2 ... A_n}$ with $3\leq m \leq n-1$.
\end{proof}

\section{Conclusions}\label{Sec: Conclusion}

In this paper, we have proposed SCREN as a powerful candidate to characterize the strongly monogamous property of multi-qudit systems.
We have provided a SM inequality of multi-party entanglement in terms of SCREN, and shown that the tangle-based SM inequality of multi-qubit systems can be rephrased by SCREN. We have also shown that SCREN SM inequality is still true for the counterexamples of CKW inequality in higher-dimensional systems. We have further provided an analytical proof that SCREN SM inequality is saturated by a large class of multi-qudit states, a superposition of multi-qudit generalized W-class states and vacuums.
Thus SCREN is a good alternative of tangle in characterizing strong monogamy of multi-party entanglement without any known counterexample even in higher-dimensional systems.

Noting the importance of the study on multi-party quantum entanglement,
our result can provide a rich reference for future
work on the study of entanglement in complex quantum systems.

\section*{Acknowledgments}
This research was supported by Basic Science Research Program through the National Research Foundation of Korea(NRF)
funded by the Ministry of Education, Science and Technology(NRF-2014R1A1A2056678) and Fusion Technology R\&D Center of SK Telecom.


\begin{thebibliography}{1}
\bibitem{Negativity}
The definition of negativity for a state $\rho$ was first proposed by Vidal and Werner~\cite{VidalW}
as $\frac{\|\rho^{T_B}\|_1-1}{2}$, and this
corresponds to the absolute value of the sum of negative eigenvalues of $\rho^{T_B}$.
Another definition of the negativity with a normalizing factor, $\frac{\|\rho^{T_B}\|_1-1}{d-1}$,
was also used for the states in a $d\otimes d'$ ($d\le d'$) quantum system~\cite{LCOK}.
To avoid this inconsistency, here we only use $\|\rho^{T_B}\|_1-1$ for our definition of the negativity,
which also coincides with the square-root of tangle for two-qubit pure states.
\bibitem{schtau}
This definition is also valid for a bipartite pure state with Schmidt-rank 2, that is,
\begin{equation*}
\ket{\psi}_{AB}=\sqrt{\lambda_1}\ket{e_0}_A\otimes\ket{f_0}_B+\sqrt{\lambda_2}\ket{e_1}_A\otimes\ket{f_1}_B,
\label{schmidt2}
\end{equation*}
for some orthonormal bases $\{\ket{e_0}_A, \ket{e_1}_A\}$ and $\{\ket{f_0}_A, \ket{f_1}_A\}$ of subsystems $A$ and $B$,
respectively.
\bibitem{omit}
For simplicity, we sometimes omit the superscript $m$ in the components of the
index vector $\vec{j}^m=(j^m_1,\ldots,j^m_{m-1})$ if there is no confusion.
\end{thebibliography}
\end{document}